\documentclass[a4paper,UKenglish,cleveref,autoref]{lipics-v2019}

\usepackage[algoruled,lined,longend,linesnumbered,commentsnumbered]{algorithm2e}
\SetKwProg{Func}{Function}{}{end}
\usepackage{amssymb,amsmath,amsfonts,amsthm}
\usepackage{booktabs}
\usepackage{color}
\usepackage{graphicx}
\usepackage{hyperref}
\usepackage{listings}
\usepackage{multicol}
\usepackage{multirow}
\usepackage{subcaption}
\usepackage{wrapfig}
\usepackage{xspace}
\usepackage{wasysym}

\renewcommand\paragraph[1]{\vspace{0.2em}\noindent\textbf{#1.}}

\def\MONIC{\textsc{Monic}\xspace}
\def\MOVEC{\textsc{Movec}\xspace}
\newcommand*{\buchi}{B\"{u}chi\xspace}

\newcommand*{\X}{\emph{\textbf{X}}}
\newcommand*{\U}{\emph{\textbf{U}}}
\newcommand*{\R}{\emph{\textbf{R}}}
\newcommand*{\F}{\emph{\textbf{F}}}
\newcommand*{\G}{\emph{\textbf{G}}}

\newcommand\POSIMON{M_\top}
\newcommand\NEGAMON{M_\bot}
\newcommand\NEUTMON{M_+}

\newcommand\NONMON{M_\times}

\newcommand\vtom[1]{\texttt{vtom}(#1)}

\title{Four-valued monitorability of $\omega$-regular languages} 

\titlerunning{Four-valued monitorability of $\omega$-regular languages}

\author{Zhe Chen}
{College of Computer Sci. and Tech., 
Nanjing University of Aeronautics and Astronautics, China}
{zhechen@nuaa.edu.cn}
{}
{This work is supported by ...}

\author{Yunyun Chen}
{College of Computer Sci. and Tech., 
Nanjing University of Aeronautics and Astronautics, China}
{}
{}
{}

\author{Robert M. Hierons}
{Department of Computer Science, 
The University of Sheffield, UK}
{r.hierons@sheffield.ac.uk}
{}
{This work is supported by ...}

\author{Yifan Wu}
{College of Computer Sci. and Tech., 
Nanjing University of Aeronautics and Astronautics, China}
{}
{}
{}

\authorrunning{Z. Chen, Y. Chen, R.\,M. Hierons and Y. Wu}

\Copyright{Zhe Chen, Yunyun Chen, Robert M. Hierons and Yifan Wu}



\category{}

\relatedversion{}

\supplement{}

\nolinenumbers 

\hideLIPIcs  

\begin{document}

\maketitle

\begin{abstract}
Runtime Verification (RV) is a lightweight formal technique in which program or system execution is monitored and analyzed, to check whether certain properties are satisfied or violated after a finite number of steps.
The use of RV has led to interest in deciding whether a property is \emph{monitorable}: whether it is always possible for the satisfaction or violation of the property to be determined after a finite future continuation.
However, classical two-valued monitorability suffers from two inherent limitations.
First, a property can only be evaluated as monitorable or non-monitorable; no information is available regarding whether only one verdict (satisfaction or violation) can be detected.
As a result, the developer may write unnecessary handlers for inactive verdicts, increasing effort and runtime overhead.
Second, monitorability is defined at the \emph{language-level} and does not tell us whether satisfaction or violation can be detected starting from the current monitor \emph{state} during system execution.
As a result, every monitor object must be maintained during the entire execution, again increasing runtime overhead.

To address these limitations, this paper proposes a new notion of four-valued monitorability for $\omega$-languages and applies it at the state-level.
Four-valued monitorability is more informative than two-valued monitorability as a property can be evaluated as a four-valued result, denoting that only satisfaction, only violation, or both are active for a monitorable property.
We can also compute state-level weak monitorability, i.e., whether satisfaction or violation can be detected starting from a given state in a monitor, which enables state-level optimizations of monitoring algorithms.
Based on a new six-valued semantics, we propose procedures for computing four-valued monitorability of $\omega$-regular languages, both at the language-level and at the state-level.
We have developed a new tool, \MONIC, that implements the proposed procedure for computing monitorability of LTL formulas.
We evaluated its effectiveness using a set of standard LTL formulas. Experimental results show that \MONIC can correctly, and quickly, report both two-valued and four-valued monitorability.
\end{abstract}

\newpage
\section{Introduction}
\label{Sec:intro}

Runtime Verification (RV) \cite{LeuS09,BarFF18,HavR18} is a lightweight formal technique in which program or system execution is monitored and analyzed.
RV uses information extracted from an execution to check whether certain properties are satisfied or violated after a finite number of steps, possibly leading to online responses, such as signaling notifications or alarms, logging, computing statistical information, profiling, and performing error protection or recovery.
In RV, properties are usually expressed using formalisms~\cite{HavR17} such as Linear Temporal Logic (LTL) formulas~\cite{Pnu77,ManP92,Dru00,BauLS11}, Nondeterministic \buchi Automata (NBAs), and $\omega$-regular expressions, which represent $\omega$-regular languages~\cite{AmoR05,Bau10}.
RV tools automatically synthesize monitors (i.e.,~code fragments) from formal specifications and then weave the code into the system through instrumentation~\cite{Gei01,HavR02,Hav08}.
The inserted code typically maintains a set of monitor objects that can detect property satisfaction or violation during system execution.
Such approaches have been extended to parametric RV, in which properties are checked over every parameter instance (i.e., a combination of parameter values) by maintaining a monitor object for every parameter instance~\cite{ChenR07,MerJGCR12,RosC12,ChenWZX16,Chen17,HavRT18}.

Figure \ref{Fig:ex_monitor_spec} shows a monitor specification, written in the \MOVEC language \cite{ChenWZX16}, for the parametric RV of an event-driven system that dispatches a variety of events (e.g., sensor status, keystrokes, program loadings etc.) to components (e.g., libraries, mobile apps, microservices etc.).
Similar specifications can be written for other tools such as JavaMOP~\cite{ChenR07,MerJGCR12} and TraceMatches~\cite{AllAC05,AvgTM07}.
This specification defines a parametric monitor, named \verb|priority|, which takes two parameters: a component ID \verb|c| and an event ID \verb|e| that should be instantiated with the values (i.e., actual arguments) generated by system execution.
The specification body begins with four actions, which extract information regarding function calls that occur during runtime:
\verb|r| records a component being registered to an event (it also creates a monitor object by instantiating the monitor parameters with the arguments of the call),
\verb|u| records an unregister,
\verb|b| records the broadcast of an event (the argument of the call) to all components,
and \verb|n| records a certain component being notified of a specific event.
This specification is used to monitor system execution to check whether the property, specified as LTL formula $\phi_1 := (r \wedge \F u) \rightarrow ((\neg b \wedge \neg u) \U n) \U u$, is satisfied or violated after a finite number of steps, i.e., any infinite future continuation makes the property satisfied or violated, respectively.
The property requires that if a component \verb|c| registers to an event \verb|e| and unregisters later, then before the unregister, the event \verb|e| cannot be broadcasted until \verb|c| has been notified (i.e., \verb|c| has a higher priority than unregistered components in the context of receiving \verb|e|).

In practice, if the satisfaction or violation of a property is detected by a monitor object then an associated handler (i.e., a piece of code) is automatically triggered to perform some online response~\cite{ChenWZX16,ChenR07,MerJGCR12}.
For example, Figure \ref{Fig:ex_monitor_spec} includes two handlers for the satisfaction (i.e., validation) and violation of the LTL formula: if the property is satisfied then a message is logged; if it is violated then an alarm is signaled and this prints the IDs of the component and the event. The two handlers may also be extended to more advanced operations, e.g., profiling and error recovery.

\begin{figure}[htb]
\small
\begin{verbatim}
monitor priority(c,e) {
  creation action r(c,e) after call(% reg_component(% %:c, % %:e));
  action u(c,e) after call(% unreg_component(% %:c, % %:e));
  action b(e) before execution(% broadcast(% %:e)); 
  action n(c,e) after execution(% notify(% %:c, % %:e));

  ltl: (r && <>u) -> ((!b && !u) U n) U u;
  @validation {
    log("Priority applied: component %lu registers to event %lu.\n",
        monitor->c, monitor->e);
  }
  @violation {
    printf("Priority violated: component %lu registers to event %lu.\n",
        monitor->c, monitor->e);
  }
};
\end{verbatim}
\vspace{-1em}
\caption{A monitor specification with an LTL formula.}
\vspace{-1em}
\label{Fig:ex_monitor_spec}
\end{figure}

We may also monitor the system against other properties, e.g., $\phi_2 := \F r \rightarrow \G \F n$ that a component should receive notifications infinitely often after its registration, $\phi_3 := r \rightarrow \F u$ that a component unregisters after its registration, and $\phi_4 := \G (r \rightarrow \neg u \U n)$ that a registered component receives at least one notification before its deregistration. The developer may also write handlers for the satisfaction and violation of each property.

When specifying properties, the developer is usually concerned with their monitorability~\cite{PnuZ06,BauLS11,Bau10,DieL14}, i.e., after any number of steps, whether the satisfaction or violation of the monitored property can still be detected after a finite future continuation.
When writing handlers for these properties, the developer might consider the following question: \emph{``Can the handlers for satisfaction and violation be triggered during system execution?''}
We say that a verdict and its handler are \emph{active} if there is some continuation that would lead to the verdict being detected and thus its handler being triggered.
This question can be partly answered by deciding monitorability (with the traditional two-valued notion).
For example, $\phi_2$ (above) is non-monitorable, i.e., there is some finite sequence of steps after which no verdict is active.
Worse, $\phi_2$ is also weakly non-monitorable \cite{ChenWW18}, i.e., no verdict can be detected after any number of steps.
Thus writing handlers for $\phi_2$ is a waste of time as they will never be triggered. More seriously, monitoring $\phi_2$ at runtime adds no value but increases runtime overhead. In contrast, $\phi_1$, $\phi_3$ and $\phi_4$ are monitorable, i.e., some verdicts are always active. Thus their handlers must be developed as they may be triggered.
However, this answer is still unsatisfactory, as the existing notion of monitorability suffers from two inherent limitations: \emph{limited informativeness} and \emph{coarse granularity}.

\paragraph{Limited informativeness}
The existing notion of monitorability is not sufficiently informative, as it is two-valued, i.e., a property can only be evaluated as monitorable or non-monitorable. This means, for a monitorable property, we only know that some verdicts are active, but no information is available regarding whether only one verdict (satisfaction or violation) is active.
As a result, the developer may still write unnecessary handlers for inactive verdicts.
For example, $\phi_1$, $\phi_3$ and $\phi_4$ are monitorable. We only know that at least one of satisfaction and violation is active, but this does not tell us which ones are active and thus which handlers are required.
As a result, the developer may waste time in handling inactive verdicts, e.g., the violation of $\phi_3$ and the satisfaction of $\phi_4$.
Thus, the existing answer is far from satisfactory.

Limited informativeness also weakens the support for property debugging.
For example, when writing a property the developer may expect that both verdicts are active but a mistake may lead to only one verdict being active. The converse is also the case.
Unfortunately, these kinds of errors cannot be revealed by two-valued monitorability, as the expected property and the written (erroneous) property are both monitorable.
For example, the developer may write formula $\phi_4$ while having in mind another one $\phi_5 := r \rightarrow \neg u \U n$, i.e., what she/he really wants is wrongly prefixed by one $\G$. These two formulas cannot be discriminated by deciding two-valued monitorability as both are monitorable.

\paragraph{Coarse granularity}
The existing notion of monitorability is defined at the language-level, i.e., a property can only be evaluated as monitorable or non-monitorable as a whole, rather than a notion for (more fine-grained) states in a monitor.
This means that we do not know whether satisfaction or violation can be detected \emph{starting from the current state} during system execution.
As a result, every monitor object must be maintained during the entire execution, again increasing runtime overhead.
For example, $\phi_6 := \G \F r \vee (\neg n \rightarrow \X \neg b)$ is weakly monitorable, thus all its monitor objects (i.e., instances of the Finite State Machine (FSM) in Figure \ref{Fig:ex_monitor}), created for every pair of component and event, are maintained.

\begin{wrapfigure}{r}{0.4\textwidth}
\vspace{-3em}\hspace{-2em}
\includegraphics[scale=0.5]{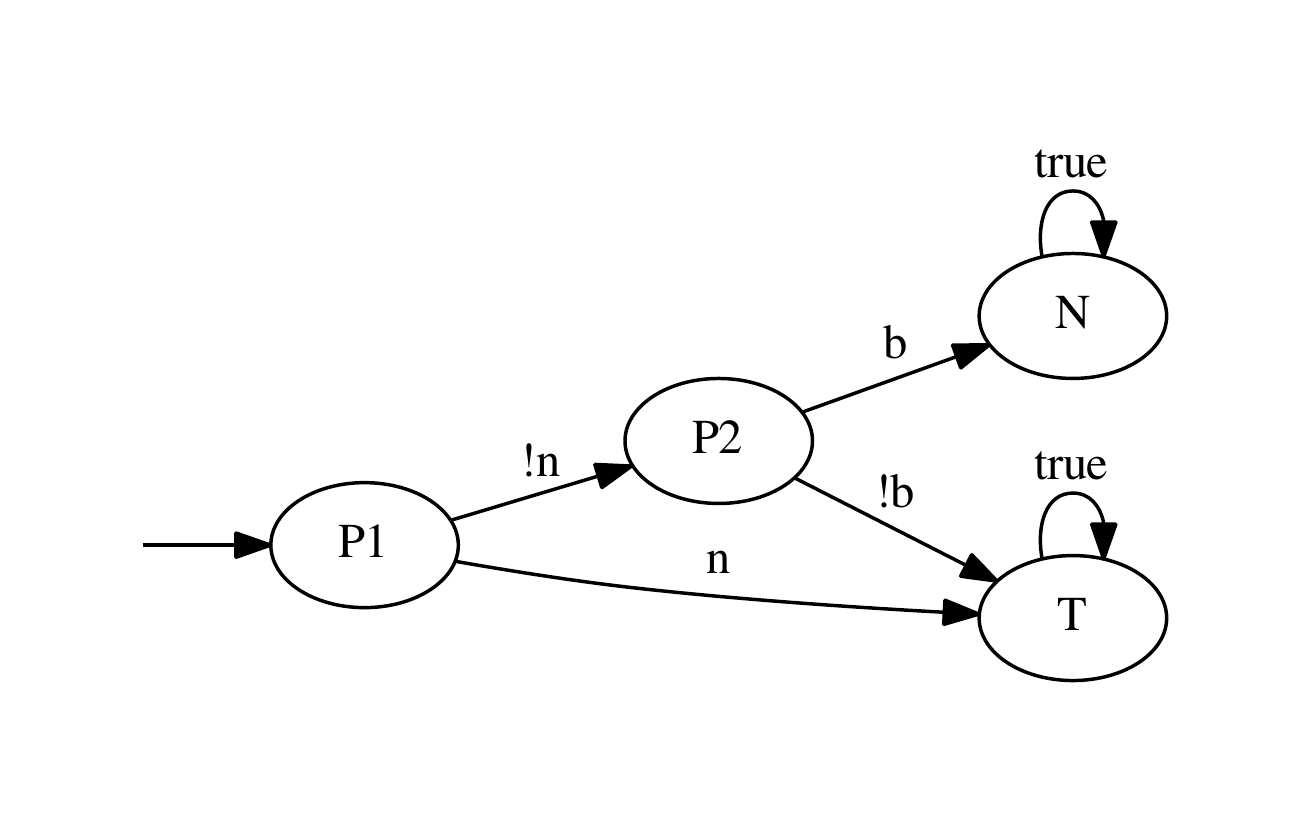}
\vspace{-4em}
\caption{A monitor for LTL formula $\phi_6 := \G \F r \vee (\neg n \rightarrow \X \neg b)$. Each transition is labeled with a propositional formula denoting a set of satisfying states. For example, ``!n'' denotes $\{ \emptyset, \{r\}, \{b\}, \{r,b\} \}$ and ``true'' denotes all states.}
\vspace{-1em}
\label{Fig:ex_monitor}
\end{wrapfigure}

Note that parametric runtime verification is NP-complete for detecting violations and coNP-complete for ensuring satisfaction~\cite{Chen17}. This high complexity primarily comes from the large number of monitor objects maintained for all parameter instances~\cite{MerJGCR12,ChenWZX16,Chen17}.
For state-level optimizations of monitoring algorithms, if no verdict can be detected starting from the current state of a monitor object, then the object can be switched off and safely removed to improve runtime performance.
For example, in Figure \ref{Fig:ex_monitor}, only satisfaction can be detected starting from states \verb|P1|, \verb|P2| and \verb|T|, whereas no verdict can be detected starting from \verb|N|. Thus a monitor object can be safely removed when it enters \verb|N|. Unfortunately, the existing notion does not support such optimizations.
 
\paragraph{Our Solution}
In this paper, we propose a new notion of four-valued monitorability for $\omega$-languages, and apply it at the state-level, overcoming the two limitations discussed above.
First, the proposed approach is more informative than two-valued monitorability. Indeed, a property can be evaluated as a four-valued result, denoting that \emph{only satisfaction, only violation, or both are active for a monitorable property}.
Thus, if satisfaction (resp. violation) is inactive, then writing handlers for satisfaction (resp. violation) is not required.
This can also enhance property debugging. For example, $\phi_4$ and $\phi_5$ can now be discriminated by their different monitorability results, as $\phi_4$ can never be satisfied but $\phi_5$ can be satisfied and can also be violated.
Thus, additional developer mistakes can be revealed.
Second, we can compute state-level weak monitorability, i.e., whether satisfaction or violation can be detected starting from a given state in a monitor.
For example, in Figure \ref{Fig:ex_monitor}, state \verb|N| is weakly non-monitorable, thus a monitor object can be safely removed when it enters state \verb|N|, which achieves a state-level optimization.

In summary, we make the following contributions.
\begin{itemize}
\item We propose a new notion of four-valued monitorability for $\omega$-languages (Section \ref{Sec:mon_4}), which provides more informative answers as to which verdicts are active.
This notion is defined using six types of prefixes, which complete the classification of finite sequences.
\item We propose a procedure for computing four-valued monitorability of $\omega$-regular languages, given in terms of LTL formulas, NBAs or $\omega$-regular expressions (Section \ref{Sec:mon_4:algorithm}), based on a new six-valued semantics.
\item We propose a new notion of state-level four-valued weak monitorability and its computation procedure for $\omega$-regular languages (Section \ref{Sec:mon_state}), which describes which verdicts are active for a state. This can enable state-level optimizations of monitoring algorithms.
\item We have developed a new tool, \MONIC, that implements the proposed procedure for computing monitorability of LTL formulas. We evaluated its effectiveness using a set of 97 LTL patterns and formulas $\phi_1$ to $\phi_6$ (above). Experimental results show that \MONIC can correctly report both two-valued and four-valued monitorability (Section \ref{Sec:experiments}).
\end{itemize}

\section{Preliminaries}

Let $AP$ be a non-empty finite set of \emph{atomic propositions}. A \emph{state} is a complete assignment of truth values to the propositions in $AP$. Let $\Sigma = 2^{AP}$ be a finite \emph{alphabet}, i.e., the set of all states. $\Sigma^*$ is the set of finite words (i.e., sequences of states in $\Sigma$), including the empty word $\epsilon$, and $\Sigma^\omega$ is the set of infinite words.
We denote atomic propositions by $p$, $q$, $r$, finite words by $u$, $v$, and infinite words by $w$, unless explicitly specified. We write a finite or infinite word in the form $\{p, q\}\{p\}\{q,r\}\cdots$, where a proposition appears in a state iff it is assigned true. We drop the brackets around singletons, i.e., $\{p, q\}p\{q,r\}\cdots$.

An \emph{$\omega$-language} (i.e., a linear-time infinitary property) $L$ is a set of infinite words over $\Sigma$, i.e., $L \subseteq \Sigma^\omega$.
Linear Temporal Logic (LTL) \cite{Pnu77,ManP92} is a typical representation of $\omega$-regular languages. LTL extends propositional logic, which uses \emph{boolean connectives} $\neg$ (not) and $\wedge$ (conjunction), by introducing \emph{temporal connectives} such as $\X$ (next), $\U$ (until), $\R$ (release), $\F$ (future, or eventually) and $\G$ (globally, or always).
Intuitively, $\X \phi$ says that $\phi$ holds at the next state, $\phi_1 \U \phi_2$ says that at some future state $\phi_2$ holds and before that state $\phi_1$ always holds.
Using the temporal connectives $\X$ and $\U$, the full power of LTL is obtained. For convenience, we also use some common abbreviations: \emph{true}, \emph{false}, standard boolean connectives $\phi_1 \vee \phi_2 \equiv \neg (\neg \phi_1 \wedge \neg \phi_2)$ and $\phi_1 \rightarrow \phi_2 \equiv \neg \phi_1 \vee \phi_2$, and additional temporal connectives $\phi_1 \R \phi_2 \equiv \neg (\neg \phi_1 \U \neg \phi_2)$ (the dual to $\U$), $\F \phi \equiv true \U \phi$ ($\phi$ eventually holds), and $\G \phi \equiv \neg \F \neg \phi$ ($\phi$ always holds).
We denote by $L(\phi)$ the $\omega$-language accepted by a formula $\phi$.

Let us recall the classification of prefixes that are used to define the three-valued semantics and two-valued monitorability of $\omega$-languages.
\begin{definition}[Good, bad and ugly prefixes \cite{KupV01, BauLS07}]
A finite word $u \in \Sigma^*$ is
a \emph{good prefix} for $L$ if $\forall w \in \Sigma^\omega.uw \in L$,
a \emph{bad prefix} for $L$ if $\forall w \in \Sigma^\omega.uw \not\in L$, or
an \emph{ugly prefix} for $L$ if no finite extension makes it good or bad, i.e., $\not\exists v\in \Sigma^*. \forall w \in \Sigma^\omega.uvw \in L$ and $\not\exists v\in \Sigma^*. \forall w \in \Sigma^\omega.uvw \not\in L$.
\end{definition}
In other words, good and bad prefixes \emph{satisfy} and \emph{violate} an $\omega$-language in some finite number of steps, respectively.
We denote by $good(L)$, $bad(L)$ and $ugly(L)$ the set of good, bad and ugly prefixes for $L$, respectively.
Note that they do not constitute a complete classification of finite words. For example, any finite word of the form $p \cdots p$ is neither a good nor a bad prefix for $p \U q$, and also is not an ugly prefix as it can be extended to a good prefix (ended with $q$) or a bad prefix (ended with $\emptyset$).

\begin{definition}[Three-valued semantics \cite{BauLS11}]
\label{Def:val_3}
Let $\mathbb{B}_3$ be the set of three truth values: \emph{true} $\top$, \emph{false} $\bot$ and \emph{inconclusive} $?$.
The truth value of an $\omega$-language $L \subseteq \Sigma^\omega$ wrt. a finite word $u \in \Sigma^*$, denoted by $[u \models L]_3$, is $\top$ or $\bot$ if $u$ is a good or bad prefix for $L$, respectively, and $?$ otherwise.
\end{definition}
Note that the inconclusive value does not correspond to ugly prefixes. Although an ugly prefix always leads to the inconclusive value, the converse does not hold. For example, $[p \cdots p \models L(p \U q)]_3$ = $?$ but $p \cdots p$ is not an ugly prefix.

Bauer et al. \cite{BauLS11} presented a monitor construction procedure that transforms an LTL formula $\phi$ into a three-valued monitor, i.e., a deterministic FSM that contains $\top$, $\bot$ and $?$ states, which output $\top$, $\bot$ and $?$ after reading over good, bad and other prefixes respectively.
For example, in Figure \ref{Fig:ex_monitor}, state \verb|T| is a $\top$ state, whereas the remaining states are all $?$ states.

The construction procedure first creates two NBAs for $\phi$ and $\neg \phi$. For each NBA, a state $q$ is marked by $\top$ if the language of the NBA starting in $q$ is not empty. The two NBAs are then converted into two Nondeterministic Finite Automata~(NFAs) whose accepting states are those marked by $\top$. The two NFAs are subsequently converted into two equivalent Deterministic Finite Automata~(DFAs) $A_{\phi}$ and $A_{\neg \phi}$ using a standard determinization, e.g., the power-set construction. The procedure finally obtains the deterministic FSM by computing and minimizing the product of the two DFAs. A state of the FSM outputs $\top$ if it does not contain an accepting state of $A_{\neg \phi}$, $\bot$ if it does not contain an accepting state of $A_{\phi}$, or $?$ otherwise. Note that each NBA is exponentially larger than the corresponding formula, and each DFA is exponentially larger than the corresponding NBA. Thus, this construction procedure requires 2ExpSpace.
This construction procedure can be adapted to construct monitors for the $\omega$-regular languages specified as NBAs and $\omega$-regular expressions with the same complexity.

It has been shown that the three-valued monitor can be used to compute the truth value of an $\omega$-language wrt. a finite word \cite{BauLS11}, which is the output of the corresponding monitor after reading over this word.
\begin{lemma}
\label{Lem:monitor_3}
Let $M = (Q$, $\Sigma$, $\delta$, $q_0$, $\mathbb{B}_3$, $\lambda_3)$ be a three-valued monitor for an $\omega$-language $L \subseteq \Sigma^\omega$, where $Q$ is a finite set of \emph{states}, $\Sigma$ is a finite \emph{alphabet}, $\delta: Q \times \Sigma \mapsto Q$ is a \emph{transition function}, $q_0 \in Q$ is an \emph{initial state}, $\mathbb{B}_3$ is an \emph{output alphabet} and $\lambda_3: Q \rightarrow \mathbb{B}_3$ is an \emph{output function}.
For any $u \in \Sigma^*$, $[u \models L]_3 = \lambda_3(\delta(q_0, u))$.
\end{lemma}

\begin{definition}[Two-valued monitorability \cite{PnuZ06,BauLS11,Bau10}]
\label{Def:mon}
An $\omega$-language $L \subseteq \Sigma^\omega$ is
\begin{itemize}
\item \emph{$u$-monitorable} for $u \in \Sigma^*$, if $\exists v \in \Sigma^*$, s.t. $uv$ is a good or bad prefix.
\item \emph{monitorable} if it is $u$-monitorable for every $u \in \Sigma^*$.
\end{itemize}
\end{definition}
In other words, $L$ is \emph{$u$-monitorable} if $u$ has a \emph{good} or \emph{bad extension}. $L$ is \emph{monitorable} if every finite word has a good or bad extension. 
Note that an ugly prefix can never be extended to a good or bad prefix. Thus, $L$ is \emph{non-monitorable} iff there exists an ugly prefix for $L$.

\section{Four-valued monitorability}
\label{Sec:mon_4}

In this section, we propose a new notion of four-valued monitorability, to provide more informative answers to monitorability checking. As we promised, it can indicate whether only satisfaction, only violation, or both are active for a monitorable property.
Two-valued monitorability cannot achieve this because its definition only requires that all finite words (i.e., $u$ in Definition \ref{Def:mon}) can be extended to good or bad prefixes (which witness satisfaction or violation, respectively), but does not discriminate them on the types and number of the verdicts that the extensions of each finite word can witness.
To address this limitation, our approach aims to discriminate accordingly these finite words by inspecting which types of prefixes they can be extended to.

To achieve this objective, we first need to propose a new classification of prefixes, as the traditional classification (as the good, the bad and the ugly) is not satisfactory due to incompleteness, i.e., it does not include the finite words that are neither good nor bad but can be extended to good or bad prefixes. Thus we introduce the notions of positive, negative and neutral prefixes, in addition to good, bad and ugly prefixes, to complete the classification.

\begin{definition}[Positive, negative and neutral prefixes]
A finite word $u \in \Sigma^*$ is a
\begin{itemize}
\item \emph{positive prefix} for $L$ if it is not good, but some finite extension makes it good but never bad, i.e., $\exists w \in \Sigma^\omega. uw \not\in L$, $\exists v \in \Sigma^*. \forall w \in \Sigma^\omega. uvw \in L$, and $\not\exists v \in \Sigma^*. \forall w\in \Sigma^\omega. uvw \not\in L$,
\item \emph{negative prefix} for $L$ if it is not bad, but some finite extension makes it bad but never good, i.e., $\exists w \in \Sigma^\omega. uw \in L$, $\exists v \in \Sigma^*. \forall w \in \Sigma^\omega. uvw \not\in L$, and $\not\exists v \in \Sigma^*. \forall w\in \Sigma^\omega. uvw \in L$, or
\item \emph{neutral prefix} for $L$ if some finite extension makes it good and some makes it bad, i.e., $\exists v\in \Sigma^*. \forall w \in \Sigma^\omega.uvw \in L$ and $\exists v\in \Sigma^*. \forall w \in \Sigma^\omega.uvw \not\in L$.
\end{itemize}
\end{definition}
We denote by $posi(L)$, $nega(L)$ and $neut(L)$ the set of positive, negative and neutral prefixes for $L$, respectively. It is easy to see that the three new sets of prefixes and the three traditional sets of good, bad and ugly prefixes are mutually disjoint.
An interesting fact, as shown by the following theorem, is that the six sets of prefixes exactly constitute the complete set of finite words.
Furthermore, the six types of prefixes directly correspond to the six-valued semantics (cf. Definition \ref{Def:val_6}). This completes the classification of prefixes.
\begin{theorem}
$good(L) \cup bad(L) \cup posi({L}) \cup nega({L}) \cup neut({L}) \cup ugly(L) = \Sigma^*$.
\end{theorem}

The traditional three-valued semantics can identify only good and bad prefixes with the truth values $\top$ and $\bot$ respectively, whereas all the prefixes of the other four types are given the same value $?$. To discriminate them, we further divide the value $?$ into four truth values.

\begin{definition}[Six-valued semantics]
\label{Def:val_6}
Let $\mathbb{B}_6$ be the set of six truth values: \emph{true} $\top$, \emph{false} $\bot$, \emph{possibly true} $\mp$, \emph{possibly false} $\pm$, \emph{possibly conclusive} $+$ and \emph{inconclusive} $\times$.
The truth value of an $\omega$-language $L \subseteq \Sigma^*$ wrt. a finite word $u \in \Sigma^*$, denoted by $[u \models L]_6$, is $\top$, $\bot$, $\mp$, $\pm$, $+$ or $\times$ if $u$ is a good, bad, positive, negative, neutral or ugly prefix for $L$, respectively.
\end{definition}
Note that the six-valued semantics models a rigorous correspondence between truth values and prefix types. Unlike the three-valued semantics, the inconclusive value now exactly corresponds to ugly prefixes.

The six-valued semantics is closely related to the types of extensions. That is, an $\omega$-language is evaluated wrt. a finite word, which is neither good nor bad, as \emph{possibly true} if the word can be extended to a good prefix but never a bad prefix, \emph{possibly false} if the word can be extended to a bad prefix but never a good prefix, \emph{possibly conclusive} if the word can be extended to both good and bad prefixes, or \emph{inconclusive} if neither.
Thus, the six-valued semantics can be used to evaluate whether a finite word can be extended to witness the satisfaction and violation of an $\omega$-language: Satisfaction is possible iff the $\omega$-language is evaluated wrt. the finite word as true, possibly true or possibly conclusive, while violation is possible iff the $\omega$-language is evaluated wrt. the finite word as false, possibly false or possibly conclusive.

The definition of four-valued monitorability is built on the following notion of four-valued $u$-monitorability which is used to discriminate finite words by inspecting which types of prefixes they can be extended to.
\begin{definition}[Four-valued $u$-monitorability]
\label{Def:umon_4}
An $\omega$-language $L \subseteq \Sigma^\omega$ is
\begin{itemize}
\item \emph{weakly positively $u$-monitorable} for $u \in \Sigma^*$, if $\exists v \in \Sigma^*$, s.t. $uv$ is a good prefix.
\item \emph{weakly negatively $u$-monitorable} for $u \in \Sigma^*$, if $\exists v \in \Sigma^*$, s.t. $uv$ is a bad prefix.
\item \emph{positively $u$-monitorable} if it is weakly positively, but not weakly negatively, $u$-monitorable. ($u$ has only good extensions, thus $u$ is a good or positive prefix.)
\item \emph{negatively $u$-monitorable} if it is weakly negatively, but not weakly positively, $u$-monitorable. ($u$ has only bad extensions, thus $u$ is a bad or negative prefix.)
\item \emph{neutrally $u$-monitorable} if it is both weakly positively and weakly negatively $u$-monitorable. ($u$ has both good and bad extensions, thus $u$ is a neutral prefix.)
\item \emph{not $u$-monitorable} if it is neither weakly positively nor weakly negatively $u$-monitorable. ($u$ has neither good nor bad extension, thus $u$ is an ugly prefix.)
\end{itemize}
\end{definition}
In other words, the traditional $u$-monitorability is split into two parts, i.e., weakly positive and weakly negative $u$-monitorability.
As a result, $L$ is $u$-monitorable iff $L$ is positively, negatively or neutrally $u$-monitorable.

\begin{definition}[Four-valued monitorability]
\label{Def:mon_4}
An $\omega$-language $L \subseteq \Sigma^\omega$ is
\begin{itemize}
\item \emph{positively monitorable} if it is positively $u$-monitorable for every $u \in \Sigma^*$.
\item \emph{negatively monitorable} if it is negatively $u$-monitorable for every $u \in \Sigma^*$.
\item \emph{neutrally monitorable} if it is $u$-monitorable for every $u \in \Sigma^*$, and is neutrally $\epsilon$-monitorable for the empty word $\epsilon$.
\item \emph{non-monitorable} if it is not $u$-monitorable for some $u\in \Sigma^*$.
\end{itemize}
\end{definition}
In other words, the set of monitorable $\omega$-languages is divided into three classes, i.e., positively, negatively and neutrally monitorable ones.
Note that the definition of neutral monitorability consists of two conditions, of which the first ensures that $L$ is monitorable while the second ensures that both of satisfaction and violation can be detected after some finite sequences of steps.
We denote the four truth values (positively, negatively, neutrally and non-monitorable) by $\POSIMON$, $\NEGAMON$, $\NEUTMON$ and $\NONMON$, respectively.

We can validate that four-valued monitorability indeed provides the informativeness we require, as described in Section \ref{Sec:intro}, by showing the following theorem, that the truth values $\POSIMON$, $\NEGAMON$, and $\NEUTMON$ indicate that only satisfaction, only violation, and both can be detected after some finite sequences of steps, respectively.
This theorem can be proved by Definitions \ref{Def:mon_4} and \ref{Def:umon_4}, in which $u$ is substituted by the empty word $\epsilon$.
\begin{theorem}
If an $\omega$-language $L \subseteq \Sigma^\omega$ is
\begin{itemize}
\item positively monitorable then $\exists u \in \Sigma^*. \forall w \in \Sigma^\omega. uw \in L$ and $\not\exists u \in \Sigma^*. \forall w \in \Sigma^\omega. uw \not\in L$.
\item negatively monitorable then $\exists u \in \Sigma^*. \forall w \in \Sigma^\omega. uw \not\in L$ and $\not\exists u \in \Sigma^*. \forall w \in \Sigma^\omega. uw \in L$.
\item neutrally monitorable then $\exists u \in \Sigma^*. \forall w \in \Sigma^\omega. uw \in L$ and $\exists u \in \Sigma^*. \forall w \in \Sigma^\omega. uw \not\in L$.
\end{itemize}
\end{theorem}
\proof Let us show the first proposition.
(1) If $L$ is positively monitorable then $L$ is positively $\epsilon$-monitorable by Definition \ref{Def:mon_4}, which implies $\exists v \in \Sigma^*. \forall w \in \Sigma^\omega. vw \in L$ by Definition \ref{Def:umon_4}.
(2) Suppose $\exists u \in \Sigma^*. \forall w \in \Sigma^\omega. uw \not\in L$. Thus $L$ is not positively $\epsilon$-monitorable by Definition \ref{Def:umon_4}, which implies $L$ is not positively monitorable by Definition \ref{Def:mon_4}. A contradiction.

The second and third propositions can be shown similarly.
\qed

Let us consider some simple but essential examples regarding basic temporal connectives. More examples, such as the formulas used in Section \ref{Sec:intro}, will be considered in Section \ref{Sec:experiments}.
\begin{itemize}
\item Formula $\F p$ is positively monitorable, as any finite word can be extended to a good prefix (ended with $p$) but never a bad prefix. This means that only satisfaction, but no violation, of the property can be detected after some finite sequences of steps.
\item Formula $\G p$ is negatively monitorable, as any finite word can be extended to a bad prefix (ended with $\emptyset$) but never a good prefix. This means that only violation, but no satisfaction, of the property can be detected after some finite sequences of steps.
\item Formula $p \U q$ is neutrally monitorable, as it is monitorable and $\epsilon$ (more generally, any finite word of the form $p\cdots p$) can be extended to both a good prefix (ended with $q$) and a bad prefix (ended with $\emptyset$). This means that both of satisfaction and violation of the property can be detected after some finite sequences of steps.
\item Formula $\G \F p$ is non-monitorable, as any finite word can never be extended to a good or bad prefix, due to the infinite continuations $\emptyset\emptyset \cdots$ and $pp \cdots$ respectively. This means that neither satisfaction nor violation of the property can be detected.
\end{itemize}

\section{Computing four-valued monitorability of $\omega$-regular languages}
\label{Sec:mon_4:algorithm}

In this section, we propose a procedure for computing the four-valued monitorability of $\omega$-regular languages, based on the six-valued semantics.

The first step is a monitor construction procedure that transforms an LTL formula into a six-valued monitor, i.e., a deterministic FSM which outputs $\top$, $\bot$, $\mp$, $\pm$, $+$ and $\times$ after reading over good, bad, positive, negative, neutral and ugly prefixes respectively.
For example, in Figure \ref{Fig:ex_monitor}, states \verb|P1|, \verb|P2| and \verb|N| are all $?$ states under the three-valued semantics. After refining the output function with the six-valued semantics, states \verb|P1| and \verb|P2| become $\mp$ states, whereas state \verb|N| becomes a $\times$ state.

The construction procedure first constructs a three-valued monitor, using the traditional approach which requires 2ExpSpace~\cite{BauLS11}.
Then we refine its output function, assigning new outputs to $?$ states.
Specifically, our procedure traverses all the states in the monitor, and for each state, starts another nested traversal to check whether a $\top$ state or a $\bot$ state is reachable.
A $?$ state is assigned output
$\mp$ if $\top$ states are reachable but no $\bot$ state is,
$\pm$ if $\bot$ states are reachable but no $\top$ state is,
$+$ if both $\top$ and $\bot$ states are reachable, or
$\times$ if neither is reachable.
This refinement step can be done in polynomial time and NLSpace (using the three-valued monitor as the input).
Thus, constructing a six-valued monitor requires also 2ExpSpace.
Let us formalize the above construction procedure.
\begin{definition}
\label{Def:monitor_6}
Let $M = (Q$, $\Sigma$, $\delta$, $q_0$, $\mathbb{B}_3$, $\lambda_3)$ be a three-valued monitor for an $\omega$-language $L \subseteq \Sigma^\omega$.
The corresponding six-valued monitor $M' = (Q$, $\Sigma$, $\delta$, $q_0$, $\mathbb{B}_6$, $\lambda)$ is obtained by refining the output function $\lambda_3$ of $M$ as shown in Figure \ref{Fig:lambda_6}.
\begin{figure}[htb]
\vspace{-1em}
\begin{equation*}  
\text{for any } q \in Q, \lambda(q) = \left\{  
  \begin{array}{ll}  
  \top, & \text{if } \lambda_3(q) = \top \\  
  \bot, & \text{if } \lambda_3(q) = \bot \\  
  \mp,  & \text{if }
          \left\{
            \begin{array}{l}
            \lambda_3(q) \neq \top \\
            \exists v \in \Sigma^*.~\delta(q, v) = q' \wedge~ \lambda_3(q') = \top, \text{ and} \\
            \forall v \in \Sigma^*.~\delta(q, v) = q' \rightarrow \lambda_3(q') \neq \bot
            \end{array}
          \right. \\
  \pm,  & \text{if }
          \left\{
            \begin{array}{l}
            \lambda_3(q) \neq \bot \\
            \exists v \in \Sigma^*.~\delta(q, v) = q' \wedge~ \lambda_3(q') = \bot, \text{ and} \\
            \forall v \in \Sigma^*.~\delta(q, v) = q' \rightarrow \lambda_3(q') \neq \top
            \end{array}
          \right. \\
  +,    & \text{if }
          \left\{
            \begin{array}{l}
            \exists v \in \Sigma^*.~\delta(q, v) = q' \wedge~ \lambda_3(q') = \top, \text{ and} \\
            \exists v \in \Sigma^*.~\delta(q, v) = q' \wedge~ \lambda_3(q') = \bot
            \end{array}
          \right. \\
  \times, & \text{if }
            \left\{
            \begin{array}{l}
            \forall v \in \Sigma^*.~\delta(q, v) = q' \rightarrow \lambda_3(q') \neq \top, \text{ and} \\
            \forall v \in \Sigma^*.~\delta(q, v) = q' \rightarrow \lambda_3(q') \neq \bot
            \end{array}
            \right.
  \end{array}
\right.
\end{equation*}
\vspace{-1em}
\caption{The output function $\lambda$.}
\vspace{-1em}
\label{Fig:lambda_6}
\end{figure}
\end{definition}

We can show the following lemma, that the six-valued monitor can be used to compute the truth value of an $\omega$-language wrt. a finite word. This lemma can be proved by Definitions \ref{Def:val_6} and \ref{Def:val_3}, Lemma \ref{Lem:monitor_3} and Definition \ref{Def:monitor_6}.
\begin{lemma}
\label{Lem:monitor_6}
Let $M = (Q$, $\Sigma$, $\delta$, $q_0$, $\mathbb{B}_6$, $\lambda)$ be a six-valued monitor for an $\omega$-language $L \subseteq \Sigma^\omega$.
For any $u \in \Sigma^*$, $[u \models L]_6 = \lambda(\delta(q_0, u))$.
\end{lemma}

As a property of the six-valued monitor, the following theorem shows that each state in a monitor can be reached by exactly one type of prefixes (by Lemma \ref{Lem:monitor_6} and Definition \ref{Def:val_6}).
\begin{theorem}
\label{Thm:monitor_6_prefix}
Let $M = (Q$, $\Sigma$, $\delta$, $q_0$, $\mathbb{B}_6$, $\lambda)$ be a six-valued monitor for an $\omega$-language $L \subseteq \Sigma^\omega$.
For a state $q \in Q$, $\lambda(q)$ equals $\top$, $\bot$, $\mp$, $\pm$, $+$ or $\times$, iff it can be reached by good, bad, positive, negative, neutral or ugly prefixes, respectively.
\end{theorem}

Based on the six-valued monitor, the second step determines the four-valued monitorability of an $\omega$-language $L$ by checking whether its monitor has some specific reachable states.
The monitorability of $L$ is
$\POSIMON$ iff neither $\times$ nor $\bot$ states are reachable (thus neither $\pm$ nor $+$ states are reachable),
$\NEGAMON$ iff neither $\times$ nor $\top$ states are reachable (thus neither $\mp$ nor $+$ states are reachable),
$\NEUTMON$ iff no $\times$ state is reachable but a $+$ state is reachable (thus both $\top$ and $\bot$ states are reachable), and
$\NONMON$ iff a $\times$ state is reachable.
These rules can be formalized:
\begin{theorem}
Let $M = (Q$, $\Sigma$, $\delta$, $q_0$, $\mathbb{B}_6$, $\lambda)$ be a six-valued monitor for an $\omega$-language $L \subseteq \Sigma^\omega$.
The monitorability of $L$:
\begin{equation*}  
\eta(L) = \left\{  
  \begin{array}{ll}
  \POSIMON, &\text{iff }
            \forall u \in \Sigma^*.~\delta(q_0, u) = q' \rightarrow \lambda(q') \neq \times \wedge \lambda(q') \neq \bot \\
  \NEGAMON, &\text{iff }
            \forall u \in \Sigma^*.~\delta(q_0, u) = q' \rightarrow \lambda(q') \neq \times \wedge \lambda(q') \neq \top \\
  \NEUTMON, &\text{iff }
            \left\{
            \begin{array}{l}
            \forall u \in \Sigma^*.~\delta(q_0, u) = q' \rightarrow \lambda(q') \neq \times, \text{ and} \\
            \exists u \in \Sigma^*.~\delta(q_0, u) = q' \wedge~ \lambda(q') = +
            \end{array}
            \right. \\
  \NONMON, &\text{iff }
           \exists u \in \Sigma^*.~\delta(q_0, u) = q' \wedge~ \lambda(q') = \times
  \end{array}  
\right.
\end{equation*}
\end{theorem}
\begin{proof}
($\Leftarrow$) This implication can be proved by Lemma \ref{Lem:monitor_6} (converting $\lambda(q')$ into $[u \models L]_6$), Definitions \ref{Def:val_6} and \ref{Def:mon_4} (matching $[u \models L]_6$ to monitorability).

For example, let us show the first case. Suppose $\forall u \in \Sigma^*.$ $\delta(q_0, u) = q'$ $\rightarrow$ $\lambda(q') \neq \times$ $\wedge$ $\lambda(q') \neq \bot$. Lemma \ref{Lem:monitor_6} ensures that $\delta(q_0, u)$ must be defined for every $u$ and $[u \models L]_6$ = $ \lambda(\delta(q_0, u))$. Thus $\forall u \in \Sigma^*.$ $[u \models L]_6  \neq \times$ $\wedge$ $[u \models L]_6 \neq \bot$. Note that $\forall u \in \Sigma^*.$ $[u \models L]_6  \neq \pm$ $\wedge$ $[u \models L]_6 \neq +$ (otherwise, $\exists v \in \Sigma^*.$ $[uv \models L]_6 = \bot$ by Definition \ref{Def:val_6}, a contradiction). As a result, $\forall u \in \Sigma^*.$ $[u \models L]_6 \in \{ \top, \mp \}$. Thus $\eta(L) = \POSIMON$ by Definitions \ref{Def:val_6} and \ref{Def:mon_4}.

($\Rightarrow$) This implication follows by the converse of the above reasoning.
\end{proof}

The above checking procedure can be done in linear time and NLSpace by traversing all the states of monitor. However, note that this procedure is performed after constructing the monitor. Thus, when an $\omega$-regular language $L$ is given in terms of an LTL formula, the four-valued monitorability of $L$ can be computed in 2ExpSpace; the same complexity as for two-valued monitorability. As we will see in Section \ref{Sec:experiments}, the small size of standard LTL patterns means that four-valued monitorability can be computed in very little time

Now consider other representations of $\omega$-regular languages.
If $L$ is given in terms of an NBA, we first explicitly complement the NBA, and the rest of the procedure stays the same. However, the complement operation also involves an exponential blowup.
If $L$ is given in terms of an $\omega$-regular expression, we first build an NBA for the expression, which can be done in polynomial time, and the rest of the procedure is the same as for NBA.
Hence, independent of the concrete representation, four-valued monitorability of an $\omega$-regular language can be computed in 2ExpSpace, by using the monitor-based procedure.

\section{State-level four-valued monitorability}
\label{Sec:mon_state}

In this section, we apply four-valued monitorability at the state-level, to predict whether satisfaction and violation can be detected \emph{starting from a given state in a monitor}.
Recall that the notions of monitorability (cf. Definitions \ref{Def:mon} and \ref{Def:mon_4}) are defined using the extensions to good and bad prefixes.
However, good and bad prefixes are defined for an $\omega$-language, not for a state.
Thus such definitions cannot be directly applied at the state-level.
Instead, we define state-level monitorability using the reachability of $\top$ and $\bot$ states, which are equivalent notions to good and bad prefixes according to Theorem \ref{Thm:monitor_6_prefix}.

\begin{definition}[State-level four-valued monitorability]
\label{Def:mon_state}
Let $M = (Q, \Sigma, \delta, q_0, \mathbb{B}_6, \lambda)$ be a six-valued monitor. A state $q \in Q$ is
\begin{itemize}
\item \emph{positively monitorable} if a $\top$ state but no $\bot$ state is reachable from $q'$, for every state $q' \in Q$ reachable from $q$.
\item \emph{negatively monitorable} if a $\bot$ state but no $\top$ state is reachable from $q'$, for every state $q' \in Q$ reachable from $q$.
\item \emph{neutrally monitorable} if a $\top$ state or a $\bot$ state is reachable from $q'$, for every state $q' \in Q$ reachable from $q$, and both a $\top$ state and a $\bot$ state are reachable from $q$.
\item \emph{non-monitorable} if neither $\top$ states nor $\bot$ states are reachable from $q'$, for some state $q' \in Q$ reachable from $q$.
\end{itemize}
\end{definition}
Note that the above state-level monitorability is too strong to meet our requirements, because it places restrictions on all the states reachable from the considered state $q$.
For example, in Figure \ref{Fig:ex_monitor}, we require discriminating states \verb|P1| and \verb|P2| from state \verb|N|, as satisfaction can be detected starting from \verb|P1| and \verb|P2|, but neither satisfaction nor violation can be detected starting from \verb|N|.
However, \verb|P1|, \verb|P2| and \verb|N| are all non-monitorable as neither $\top$ states nor $\bot$ states are reachable from \verb|N|.
To provide the required distinction, we should use a weaker form of state-level monitorability as follows.

\begin{definition}[State-level four-valued weak monitorability]
\label{Def:mon_weak_state}
Let $M = (Q, \Sigma, \delta, q_0, \mathbb{B}_6, \lambda)$ be a six-valued monitor. A state $q \in Q$ is
\begin{itemize}
\item \emph{weakly positively monitorable} if a $\top$ state but no $\bot$ state is reachable from $q$.
\item \emph{weakly negatively monitorable} if a $\bot$ state but no $\top$ state is reachable from $q$.
\item \emph{weakly neutrally monitorable} if both a $\top$ state and a $\bot$ state are reachable from $q$.
\item \emph{weakly non-monitorable} if neither $\top$ states nor $\bot$ states are reachable from $q$.
\end{itemize}
A state is \emph{weakly monitorable}, iff it is weakly positively, negatively or neutrally monitorable.
\end{definition}
For example, in Figure \ref{Fig:ex_monitor}, states \verb|P1|, \verb|P2| and \verb|T| are all weakly positively monitorable as \verb|T| is a reachable $\top$ state, while state \verb|N| is weakly non-monitorable. Thus, states \verb|P1| and \verb|P2| can now be discriminated from state \verb|N|.

We can validate that state-level four-valued weak monitorability can indeed predict whether satisfaction and violation can be detected \emph{starting from a given state}, as anticipated in Section \ref{Sec:intro}, by showing the following theorem, that the truth values $\POSIMON$, $\NEGAMON$, $\NEUTMON$ and $\NONMON$ indicate that only satisfaction, only violation, both and neither can be detected, respectively.
This theorem can be proved by Definition \ref{Def:mon_weak_state} and Theorem \ref{Thm:monitor_6_prefix}.
\begin{theorem}
Let $M = (Q, \Sigma, \delta, q_0, \mathbb{B}_6, \lambda)$ be a six-valued monitor.
Suppose a state $q \in Q$ can be reached from $q_0$ by reading over $u \in \Sigma^*$, i.e., $\delta(q_0, u) = q$. If $q$ is
\begin{itemize}
\item weakly $\POSIMON$ then $\exists v \in \Sigma^*. \forall w \in \Sigma^\omega. uvw \in L$ and $\not\exists v \in \Sigma^*. \forall w \in \Sigma^\omega. uvw \not\in L$.
\item weakly $\NEGAMON$ then $\exists v \in \Sigma^*. \forall w \in \Sigma^\omega. uvw \not\in L$ and $\not\exists v \in \Sigma^*. \forall w \in \Sigma^\omega. uvw \in L$.
\item weakly $\NEUTMON$ then $\exists v \in \Sigma^*. \forall w \in \Sigma^\omega. uvw \in L$ and $\exists v \in \Sigma^*. \forall w \in \Sigma^\omega. uvw \not\in L$.
\item weakly $\NONMON$ then $\not\exists v \in \Sigma^*. \forall w \in \Sigma^\omega. uvw \in L$ and $\not\exists v \in \Sigma^*. \forall w \in \Sigma^\omega. uvw \not\in L$.
\end{itemize}
\end{theorem}

The four truth values can be used in state-level optimizations of monitoring algorithms:
\begin{itemize}
\item If a state is weakly positively (resp. negatively) monitorable, then a monitor object can be safely removed when it enters this state, provided that only violation (resp. satisfaction) handlers are specified, as no handler can be triggered.
\item If a state is weakly neutrally monitorable, then a monitor object must be preserved if it is at this state as both satisfaction and violation can be detected after some continuations.
\item If a state is weakly non-monitorable, then a monitor object can be safely removed when it enters this state as no verdict can be detected after any continuation.
\end{itemize}
Besides, a monitor object can also be removed when it enters a $\top$ state or a $\bot$ state, as any finite or infinite continuation yields the same verdict.

Let us consider the relationship between the language-level monitorability and the state-level weak monitorability. The following lemma shows that the monitorability of an $\omega$-language depends on the weak monitorability of all the reachable states of its monitor.
This means, if an $\omega$-language is non-monitorable, then the state space of its monitor may consist of both weakly monitorable and weakly non-monitorable states.
\begin{lemma}
\label{Lem:mon__mon_weak}
Let $M = (Q$, $\Sigma$, $\delta$, $q_0$, $\mathbb{B}_6$, $\lambda)$ be a six-valued monitor for an $\omega$-language $L \subseteq \Sigma^\omega$.
$L$ is monitorable iff every reachable state of $M$ is weakly monitorable.
\end{lemma}
\begin{proof}
($\Rightarrow$) Suppose $L$ is monitorable.
By Definition \ref{Def:mon}, $\forall u \in \Sigma^*$, $\exists v \in \Sigma^*$, s.t. $uv$ is a good or bad prefix.
By Definition \ref{Def:val_6}, this means $\forall u \in \Sigma^*.$ $\exists v \in \Sigma^*.$ $[uv \models L]_6 = \top$ $\vee$ $[uv \models L]_6 = \bot$.
By Lemma \ref{Lem:monitor_6}, we have $\forall u \in \Sigma^*.$ $\exists v \in \Sigma^*.$ $\lambda(\delta(q_0, uv)) = \top$ $\vee$ $\lambda(\delta(q_0, uv)) = \bot$.
Let $\delta(q_0, u) = q$.
It follows $\forall u \in \Sigma^*.$ $\delta(q_0, u) = q$ $\wedge$ $\exists v \in \Sigma^*.$ $\lambda(\delta(q, v)) = \top$ $\vee$ $\lambda(\delta(q, v)) = \bot$. That is, for every reachable state $q$, a $\top$ state or a $\bot$ state is reachable from $q$. Thus every reachable state of $M$ is weakly monitorable by Definition \ref{Def:mon_weak_state}.

($\Leftarrow$) This implication follows by the converse of the above reasoning.
\end{proof}

Let us consider how one can compute the state-level four-valued weak monitorability for each state in a six-valued monitor.
We first formalize a mapping from truth values to weak monitorability, and then show that the state-level weak monitorability can be quickly computed from the output of the state.

\begin{definition}[Value-to-weak-monitorability]
\label{Def:vtom}
Let $\texttt{vtom}: \mathbb{B}_6 \mapsto \mathbb{M}_4$ be the \emph{value-to-weak-monitorability operator} that converts a truth value in $\mathbb{B}_6$ into the corresponding result of weak monitorability in $\mathbb{M}_4 = \{ \POSIMON, \NEGAMON, \NEUTMON, \NONMON \}$, defined as follows: $\vtom{\top} = \vtom{\mp} = \POSIMON$, $\vtom{\bot} = \vtom{\pm} = \NEGAMON$, $\vtom{+} = \NEUTMON$ and $\vtom{\times} = \NONMON$.
\end{definition}

\begin{theorem}
\label{Thm:vtom_state}
Let $M = (Q$, $\Sigma$, $\delta$, $q_0$, $\mathbb{B}_6$, $\lambda)$ be a six-valued monitor for an $\omega$-language $L \subseteq \Sigma^\omega$.
The four-valued weak monitorability of $q \in Q$ equals $\vtom{\lambda(q)}$.
\end{theorem}
\begin{proof}
On the one hand, consider all possible values of $\lambda(q)$ in $\mathbb{B}_6$.
For each possible value, we can find the corresponding type of prefixes reaching $q$ by Theorem \ref{Thm:monitor_6_prefix}. We know whether these prefixes can be extended to good and bad prefixes, i.e., to reach $\top$ and $\bot$ states, again by Theorem \ref{Thm:monitor_6_prefix}. Thus the four-valued weak monitorability of $q$ can be inferred by Definition \ref{Def:mon_weak_state}.
On the other hand, $\vtom{\lambda(q)}$ can be computed by Definition \ref{Def:vtom}.
It is easy to validate that they are always equal.
\end{proof}

\section{Implementation and experimental results}
\label{Sec:experiments}

We have developed a new tool, \MONIC, that implements the proposed procedure for computing four-valued monitorability of LTL formulas. \MONIC also supports deciding two-valued monitorability.
We have evaluated its effectiveness using a set of LTL formulas, including formulas $\phi_1$ to $\phi_6$ (used in Section \ref{Sec:intro}) and Dwyer et al.'s 97 LTL patterns \cite{DwyAC99,BauLS11}.
The evaluation was performed on an ordinary laptop, equipped with an Intel Core i7-6500U CPU (at 2.5GHz), 4GB RAM and Ubuntu Desktop (64-bit).

The result on formulas $\phi_1$ to $\phi_6$ shows that:
$\phi_1$ is neutrally monitorable,
$\phi_2$ is non-monitorable,
$\phi_3$ is positively monitorable,
$\phi_4$ is negatively monitorable,
$\phi_5$ is neutrally monitorable, and
$\phi_6$ is non-monitorable (but weakly monitorable).
Thus, the violation of $\phi_3$ and the satisfaction of $\phi_4$ can never be detected, whereas both verdicts are active for $\phi_1$ and $\phi_5$.
Further, $\phi_4$ and $\phi_5$ can be discriminated by their different monitorability results.

We also ran \MONIC on Dwyer et al.'s specification patterns \cite{DwyAC99,BauLS11}, of which 97 are well-formed LTL formulas.
The result shows that 55 formulas are monitorable and 42 are non-monitorable.
For those monitorable ones, 6 are positively monitorable, 40 are negatively monitorable and 9 are neutrally monitorable.
Our result disagrees with the two-valued result reported in \cite{BauLS11} only on the 6th LTL formula listed in the Appendix of \cite{BauLS11}. More precisely, \MONIC reports negatively monitorable, whereas the result in \cite{BauLS11} is non-monitorable.
The formula is as follows (\verb|!| for $\neg$, \verb|&| for $\wedge$, \verb&|& for $\vee$, \verb|->| for $\rightarrow$, \verb|U| for $\U$, \verb|<>| for $\F$, \verb|[]| for $\G$):
\begin{verbatim}
[](("call" & <>"open") ->
   ((!"atfloor" & !"open") U
    ("open" | (("atfloor" & !"open") U
               ("open" | ((!"atfloor" & !"open") U
                          ("open" | (("atfloor" & !"open") U
                          ("open" | (!"atfloor" U "open")))))))))) 
\end{verbatim}
The result in \cite{BauLS11} is unreliable as it is based on manual inspection of monitors and no tool is implemented in that work.
To validate, a manual inspection of its monitor (in Figure \ref{Fig:pattern_monitor}) shows that our result is correct. Indeed, state \verb|F| is a $\bot$ state, and states \verb|N1| to \verb|N7| are all $\pm$ states that can reach the $\bot$ state \verb|F|.

\begin{figure}[htb]
\centering
\includegraphics[scale=0.25]{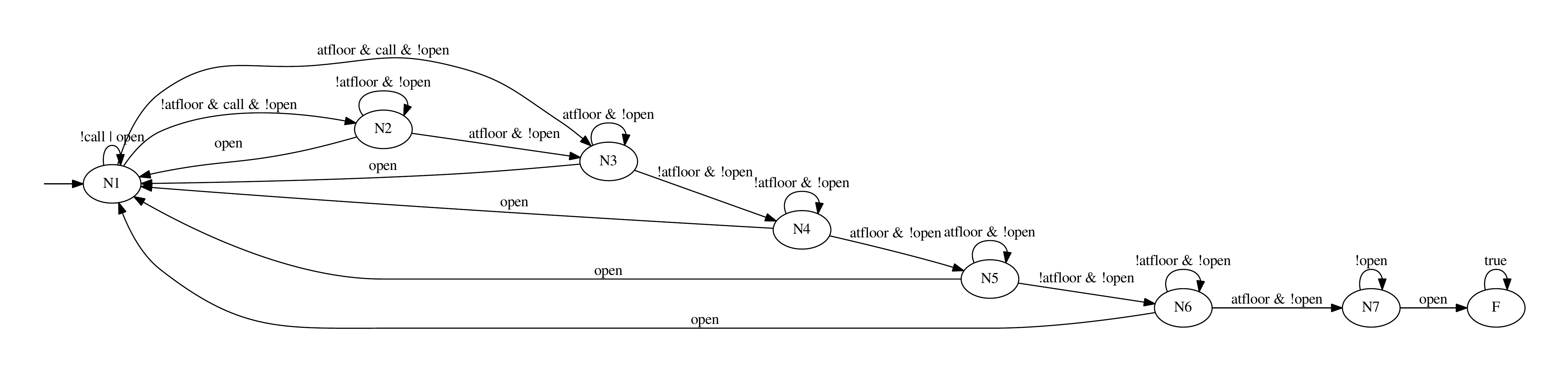}
\vspace{-1em}
\caption{The monitor of an LTL pattern.}
\label{Fig:pattern_monitor}
\end{figure}

Finally, the above results for $\phi_1$ to $\phi_6$ and the 97 LTL patterns were computed in 0.03 and 0.07 seconds, with 16MB and 20MB memory consumed, respectively (all reported by GNU \verb|time|).
To conclude, the results show that \MONIC can correctly report both two-valued and four-valued monitorability of typical formulas in very little time.

\section{Related work}

Monitorability is a principal foundational question in RV because it delineates which properties can be monitored at runtime.
The classical results on monitorability have been established for $\omega$-languages, especially for LTL \cite{PnuZ06,BauLS11,Bau10}.
Francalanza and Aceto et al. have studied monitorability for the Hennessy-Milner logic with recursion, both with a branching-time semantics \cite{Fra16,FraAA17,FraAI17,AceAF18} and with a linear-time semantics \cite{AceAF19a}.
There exist some variants of monitorability as well.
For example, monitorability has been considered over unreliable communication channels which may reorder or lose events \cite{KauHF19}.
However, all of the existing works only consider two-valued notions of  monitorability at the language-level.

Monitorability has been studied in other contexts.
For example, a topological viewpoint~\cite{DieL14} and the correspondence between monitorability and the classifications of properties (e.g., the safety-progress and safety-liveness classifications) \cite{FalFM09,FalFM12,PelH19} have been established.
A hierarchy of monitorability definitions (including monitorability and weak monitorability~\cite{ChenWW18}) has been defined wrt. the operational guarantees provided by monitors \cite{AceAF19b}.

A four-valued semantics for LTL \cite{BauLS07,BauLS10} has been proposed to refine the three-valued semantics~\cite{BauLS11}.
It divides the inconclusive truth value ? into two values: \emph{currently true} and \emph{currently false}, i.e., whether the finite sequence observed so far satisfies the property based on a finite semantics for LTL.
Note that it provides more information on what has already been seen, whereas our six-valued semantics describes what verdicts can be detected in the future continuation.

\section{Conclusion}

We have proposed four-valued monitorability and the corresponding computation procedure for $\omega$-regular languages. Then we applied the four-valued notion at the state-level.
To our knowledge, this is the first study of multi-valued monitorability, inspired by practical requirements from RV.
We believe that our work and implementation can be integrated into RV tools to provide information at the development stage and thus avoid the development of unnecessary handlers and the use of monitoring that cannot add value, enhance property debugging, and enable state-level optimizations of monitoring algorithms.



\bibliography{references}

\end{document}